\documentclass[10pt, conference, final, a4paper, twocolumn, oneside]{IEEEtran}
\ifCLASSINFOpdf \else \fi
\pdfoutput=1
\usepackage{times}
\usepackage{amssymb}
\usepackage[cmex10]{amsmath}

\usepackage{tikz}
\usetikzlibrary{shapes,arrows,shapes.geometric,fit,calc,positioning,automata,shadows}
\usepackage{pgfplots} \pgfplotsset{compat=1.7}
\usepackage[noadjust]{cite}

\usepackage[amsmath,thmmarks]{ntheorem}
\theoremseparator{.}

\theorembodyfont{\itshape}

\newtheorem{proposition}{Proposition}

\theorembodyfont{\normalfont}
\newtheorem{definition}{Definition}
\newtheorem{example}{Example}

\theoremstyle{nonumberplain}
\theoremsymbol{\ensuremath{\square}}
\theoremheaderfont{\normalfont\itshape}
\newtheorem{proof}{Proof}

\begin{document}
\title{Industrial Symbiotic Relations as Cooperative Games}
\author{\IEEEauthorblockN{Vahid Yazdanpanah and Devrim Murat Yazan}
\IEEEauthorblockA{University of Twente, Drienerlolaan 5, 7522 NB, Enschede, The Netherlands.\\ Email: \{v.yazdanpanah,d.m.yazan\}@utwente.nl}}
\IEEEspecialpapernotice{\normalsize (Presented at the $7^{th}$ IESM Conference, October 11--13, 2017, Saarbr\"{u}cken, Germany)}
\maketitle

\begin{abstract} In this paper, we introduce a game-theoretical formulation for a specific form of collaborative industrial relations called ``Industrial Symbiotic Relation (ISR) games" and provide a formal framework to model, verify, and support collaboration decisions in this new class of two-person operational games. ISR games are formalized as cooperative cost-allocation games with the aim to allocate the total ISR-related operational cost to involved industrial firms in a fair and stable manner by taking into account their contribution to the total traditional ISR-related cost. We tailor two types of allocation mechanisms using which  firms can implement cost allocations that result in a collaboration that satisfies the  fairness and stability properties. Moreover, while industries receive   a particular ISR proposal, our introduced methodology is applicable as a managerial decision support to systematically verify the quality of the ISR in question. This is achievable by analyzing if the implemented allocation mechanism is a stable/fair allocation.\end{abstract}

\section{Introduction}

The multi-dimensional concept of  \emph{Industrial Symbiosis} focuses on analysis, design, and operation of collaborative relations between traditionally disjoint industrial enterprises with the aim of keeping reusable resources, e.g., recyclable material or waste energy, in their (loosely connected) value chains \cite{chertow2000industrial,lombardi2012redefining,yazan2016design}. As \emph{Industrial Symbiotic Relations} (ISRs) aim at the lowest possible discharge of resources, they can be considered as a tool for implementing the concept of \emph{circular economy} \cite{pearce1990economics} in the context of industrial relations. Moreover, ISRs are closely related to  \emph{Industry 4.0} paradigm and practice of  \emph{Collaborative Networked Organizations} as they all are concerned about the necessity for interrelation between traditionally disconnected industrial firms \cite{stock2016opportunities,camarinha2009collaborative}. Reviewing industrial symbiosis literature, we encounter recent contributions focused on different aspects of this concept. In  \cite{yazan2016design}, they present the concept of \emph{perfect industrial symbiosis} and verify the quality of any given ISR by measuring its distance to such a perfect form. Introduced method in \cite{fraccascia2017technical} focuses on efficiency measuring while \cite{jackson1998s} studies  dynamics of profits in  ISRs. Despite contributions that discuss static (multi-criteria) decision analysis in industrial symbiosis (see \cite{yu2015makes}), one aspect of ISRs that we believe requires more attention is \emph{dynamic decision analysis}. In our view, while we shift from ISR in theory to ISR in practice, two missed elements  are 1) applicable decision analysis methods and 2) practical decision support tools that are aware of dynamic operational aspects of ISRs, e.g., methods for analyzing and mechanisms for designing fair and stable collaborations. This asks for formal frameworks  tailored to model, verify, and support such decisions (i.e., decision process modeling, decisions verification methods, and decision support tools). In a general view, decisions in ISRs can be categorized in two classes, \emph{selection decisions} and \emph{collaboration decisions}. The former is about  choosing  among  firms and learning about potential ISRs (exploration) while the latter is  about getting engaged in (or rejecting) a particular ISR proposal (exploitation). To deal with these two operational decision problems, the mature field of of cooperative game theory \cite{driessen2013cooperative} and more specifically subfield of Operations Research (OR) games \cite{borm2001operations} provide vigorous analytical methods and design mechanisms. 

In this work, we aim to fill the gap by tailoring analytical tools based on game-theoretical solution concepts to support the second form of decisions, i.e., collaboration decisions, in ISRs. For this purpose, we represent  ISRs as market games and model them as cooperative cost-allocation games (see \cite{shapley1955markets,young1985cost}). Accordingly, the focus is on operational aspects of ISRs, analysis of collaboration decisions in ISRs, and tailoring cost-allocation mechanisms that respect the operation of ISRs. Note that in this work we analyze collaboration decisions in bilateral industrial symbiotic relations as the nuclear building blocks for various industrial symbiotic topologies, e.g., Industrial Symbiotic Networks (ISNs) \cite{yazdanpanah2016}. 

The paper is structured as follows. In Section \ref{sec:analysis} we provide a conceptual analysis of ISRs from an operational point of view. Section \ref{sec:prelim} presents the game-theoretical preliminaries and our proposed class of ISR games. In Sections \ref{sec:mechanism}  we introduce the two tailored solution concepts for allocation of costs in ISRs. Using these notions, firms can reason about stability and fairness of any given ISR. Finally, concluding remarks and future work directions  are  presented in Section \ref{sec:conc}.

\section{Conceptual Analysis} \label{sec:analysis}

To discuss the intuition behind our proposed operational semantics for the game-theoretical formulation of Industrial Symbiotic Relations (ISRs) and to elaborate nuances of collaboration decisions in ISRs, we present the following running example. Imagine a  glass manufacturer firm $A$ and a ceramics manufacturer  firm $B$. Firm  $A$ produces glass powder as its excess resource $r$ that (after recycling) can be substituted with $i$, the primary input of $B$ for its production processes (Figure \ref{fig:auto01}). 

\begin{figure}[!htb]
\centering   
\begin{tikzpicture}[shorten >=1pt,node distance=3.25cm,on grid,auto] 
    \node [state, fill=white, drop shadow] (A)  {A};
    \node [state, fill=white, drop shadow] (B) [right of=A] {B}; 
    \node (s) [right =2.0cm of B,fill=white,text width=1.25cm, drop shadow]{\footnotesize Traditional $i$-Supplier};
    \node (d) [left =2.0cm of A,fill=white,text width=1.45cm, drop shadow]{\footnotesize $r$-Discharge Area};

  \path[->] 
  (A)  
    edge [pos=0.5, sloped, above] node {$r$ {\footnotesize (to substitute $i$) } }  (B)
    edge [dashed, pos=0.5, sloped, above] node {$r$} (d)
  (s)
  edge [dashed, pos=0.5, sloped, above] node {$i$} (B)

  ;
  
\end{tikzpicture} \caption{Schematic industrial symbiotic relation  between (glass manufacturer) firm $A$   and (ceramics manufacturer) firm $B$ on resource $r$ (glass powder).} \label{fig:auto01} \end{figure}
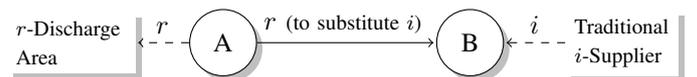

\subsection{Collaboration Decision and Cost-Allocation} \label{sec:2A}

In our ISR example, having the potential to form a symbiotic relation over $r$ may enable both $A$ and $B$ to reduce the costs affiliated with discharging excess $r$ ($7$ units of util\footnote{A \emph{util} can be  any sort of transferable utility, e.g., say that each \emph{util} is one thousand Euros.}) and purchasing $i$ ($11$ units of util), respectively.  However, establishing such a collaboration has its corresponding operational costs, e.g., costs of transportation and recycling ($15$ units of util).  Intuitively, such an ISR would be feasible if total traditional costs in case of  non-cooperation ($7+11$ utils) be more than total operational costs in case of  cooperation ($15$ utils). Such a feasibility condition  is necessary but not sufficient for firms to collaborate. One main question is about the method(s) to allocate the operational costs of collaboration such that the contribution of each party will be respected. Should both firms $A$ and $B$ equally pay ($7.5$ utils) or should firm $B$ pay more as it is enjoying more cost reduction thanks to the collaboration? Although the collaboration may reduce the total cost, how to allocate the total cost of ISR operationalization  will be at stake. This relates the practice of ISRs to the concept of \emph{coopetition} \cite{bengtsson2000coopetition} in which the players are interested in cooperation but also compete to gain as much as possible out of the benefits created by such a cooperation. Accordingly, studying if an ISR is a stable and fair relation calls for methods to analyze and mechanisms to guarantee its  performance in the operationalization stage.

\subsection{Dynamics of ISR-Related Costs}

Following \cite{albino2016exploring}, we deem that economic profitability can be seen as the main driver for industries to get involved in a potential ISR (\emph{collaboration decision}).  In brief, for firms with the potential to establish an ISR, it is reasonable to compare the total cost of ISR-operationalization with their current ISR-related costs, e.g., excess-resource discharge costs. Note that we circumscribe our focus on costs that can be reduced or costs that will be introduced, in case of ISR materialization and not the costs that remain uninfluenced, e.g., costs related to processes that are independent of the ISR in question. In the following, we provide a brief cost analysis for firms on both sides of a potential ISR. This results in two classes of ISR-related \emph{operational} and \emph{traditional} costs. The former refers to costs involved in the process of ISR operationalization while the latter is about traditional costs that firms should take into account if they do not implement the ISR. 

In our example, firm $B$ should traditionally purchase $i$ from its $i$-supplier(s) and firm $A$ has to pay the cost of discharging $r$. We call the former, \emph{traditional purchasing cost} and the latter, \emph{traditional discharge cost}. On the other hand, the three main ISR-related operational  costs  are \emph{Treatment},  \emph{Transportation},  and \emph{Transaction} costs \cite{esty1998industrial,sinding2000environmental}.  \emph{Treatment Cost}: When a resource based on which an ISR can be established, e.g., glass powder,  is out of a production process, it needs to be treated. Depending on the  resource type, treatment processes may include sorting, dismantling, liquefaction, etc.  \cite{magram2011worldwide,lovelady2009design}. Based on the set of treatment processes required to make the resource usable for the resource-receiver side of an ISR, the implementation of waste treatment facility may change. In general, the treatment process results in  a total treatment cost  for any particular ISR. \emph{Transportation Cost}:  Resource transportation can be done via land vehicles, sea freights, or even combined transportation modes (see \cite{zijm2015logistics}) with respect to the resource type and geographical boundaries. Moreover, potential partners may decide to invest in implementation of new infrastructures, e.g., a pipeline system, and paying the investment cost for this. In this work, we abstract from such subtleties in decision-making for the mode of transportation  and assume a standard total transportation cost for a given ISR. \emph{Transaction Cost}: The role of transaction costs  in establishment of ISRs is studied  in the industrial symbiosis literature, e.g., in \cite{carpenter2008use,chen2012impact}. According to \cite{dahlman1979problem,williamson1981economics}, transaction costs include the costs of market research, contracting negotiations, coordination, and adapting to the use of non-traditional resources, e.g., wastes. As discussed, the former two operational  costs are very much dependent on the resource type while the transaction cost merely depends on the administrative aspects of the ISR. As we are focused on industrial symbiotic relations (and not networks), we take into account a single value for transaction cost for a given ISR. 

In this work, we consider the two classes of ISR-related operational and traditional costs  as the main quantitative parameters for our game-theoretical formalization of ISRs. This is to tailor mechanisms for allocation of \emph{ISR-related operational costs} based on the contribution of the involved firms to \emph{ISR-related traditional costs}. For notational convenience, while discussing about a given  industrial symbiotic relation $\sigma$ between two arbitrary firms $A$ and $B$,  we denote  the total $\sigma$-related traditional costs for the firm $i\in \{A,B\}$ with $T_i(\bar{\sigma})$   and refer to total $\sigma$-related operational costs as $T(\sigma)$. So, the aim is to  allocate  $T(\sigma)$ to both $A$ and $B$ in  a stable and fair manner by taking into account $T_A(\bar{\sigma})$ and $T_B(\bar{\sigma})$. We later discuss about and distinguish between stability and fairness of cost-allocations in ISRs.

\subsection{Game-Theoretical Cost-Allocation Mechanisms} \label{sec:GTanalysis}

As we discussed in Section  \ref{sec:2A}, allocation mechanisms can play a key role in the process of ISR operationalization since a fair allocation of operational costs can foster  the collaboration.  For developing such practical allocation mechanisms, cooperative game theory \cite{driessen2013cooperative} and Operation Research (OR) games \cite{borm2001operations} provide theoretical solution concepts to allocate costs to involved players in market games. Notions such as \emph{core of the game}  and \emph{Shapley value} guarantee desired properties such that it is reasonable for firms not to deviate from cooperation \cite{driessen2013cooperative}.  In the following, we briefly analyze  properties of these two types of game-theoretical solution concepts as two notions that we aim to tailor for allocation of the total  operational cost in ISRs. 

We first discuss the concept of \emph{core of the game} as the set of all cost-allocations that (1) allocate a cost to each player lower than their traditional cost and (2) guarantee that the total allocated cost is equal to the  total operational cost. In our ISR example, the total ISR-related operational cost is $15$ utils while firms $A$  and $B$ had to traditionally pay $7$ and $11$ utils, respectively. In this case, cooperation  results in total cost reduction by $3$ utils. However, cooperating will be rational for each player, only if they individually pay less than what they had to pay traditionally. When a cost-allocation mechanism satisfies this, it  regards the so called \emph{individual rationality (INR)} property \cite{young1985cost}. On the other hand, the summation of allocated costs to players should be equal to the total operational cost. Mechanisms that satisfy this property, are called \emph{efficient (EFF)} cost-allocations. The set of cost allocations that satisfy both $INR$ and $EFF$ form the \emph{core of the game} \cite{driessen2013cooperative}. The second game-theoretic notion for allocation of costs, is the concept of \emph{Shapley value} \cite{shapley1953value} as the unique efficient (EFF) mechanism for allocation of a cost among players such that (1) symmetric players pay the same costs, (2) players that their presence in the cooperation results in no cost reduction (referred as dummy players) pay their traditional costs, and (3) if players get involved in another game, the total allocated costs to each player in these two games, can be simply added. These three properties respectively referred to  \emph{symmetry (SYM)}, \emph{dummy/null  player (DUM)}, and \emph{additivity (ADD)}  properties  in the cooperative game theory literature \cite{driessen2013cooperative}. Note that Shapley value is the unique mechanism that allocates a cost value to each player such that all $EFF$, $SYM$, $DUM$, and $ADD$ are satisfied. 

\section{Industrial Symbiotic Relations as Games} \label{sec:prelim}

Dynamics of costs and cost-saving values that result from collaboration in market games are often modeled by cooperative games with \emph{Transferable Utility} (TU) games  \cite{driessen2013cooperative,shapley1955markets}. Such games specify all the possible collaborative agent groups and represent the corresponding cost values. This formal representation  enables reasoning about cost-saving as a quantitative outcome of cooperation among agents.

\begin{definition}[Cooperative TU Games]\cite{driessen2013cooperative} A cooperative  cost-allocation game with transferable utility (a TU game) is a tuple $(N,c)$ where $N=\{a_1,a_2,\dots, a_n\}$ is the finite set of agents and $c: 2^N \mapsto \mathbb{R}_{\geq 0} $ is a characteristic cost function that associates a real number $c(S)$ with each subset $S \subseteq N$. By convention, we always assume that $c(\emptyset)=0$.\end{definition}

In the following definition, we recall two properties that axiomatize the  behavior of the cost function in response to structural relations between agent groups. 

\begin{definition} [Subadditive and Submodular Games]~\cite{driessen2013cooperative} \label{def:subs}  Let $G=(N,c)$ be a cost-allocation TU game.  We say $G$ is \emph{subadditive} iff  $c(S)+c(T) \geq c(S\cup T)$ for all disjoint agent groups $S$ and $T$ in $N$  (i.e., $S,T \subseteq N$ and $S \cap T = \emptyset$). Moreover, we say $G$ is \emph{submodular} iff $c(S)+c(T) \geq c(S\cup T) + c(S\cap T)$ for all agent groups $S$ and $T$ in $N$ (i.e., $S,T \subseteq N$). \end{definition}

Roughly speaking, in \emph{subadditive} games, agents have rational incentives to cooperate because the total cost will be higher in case of no-cooperation. In most applications, cost-allocation games are usually subadditive. Desirable properties of \emph{submodular} games (also referred as \emph{concave} games) will be elaborated while we focus on cost allocation mechanisms for TU games (in Section \ref{sec:mechanism}). 

With respect to our scope of application, i.e., bilateral symbiotic relations between industrial agents, we focus on two-person TU games  and formalize our Industrial Symbiotic Relation (ISR) games as such. Moreover, following our presented analysis in Section \ref{sec:2A} about the feasibility of ISRs (in case they result in the reduction of the total cost), we assume subadditivity as it corresponds to the nature of our application context. 

\begin{definition}[ISR Games] \label{def:ISR-Games} Let $\sigma$ be an ISR between firms $A$ and $B$. Moreover, let $T(\sigma)$ and $T_i(\bar{\sigma})$ respectively represent the total $\sigma$-related operational cost and the total $\sigma$-related traditional costs for $i \in \{A,B\}$. We say $\sigma$-based ISR cost allocation game (ISR game $\sigma$) between firms $A$ and $B$ is a subadditive cooperative TU-game $(N,c)$ where $N=\{A,B\}$, $c(N)=T(\sigma)$, $c(\emptyset)=0$, and $c(\{i\})=T_i(\bar{\sigma})$ for $i \in N$. \end{definition}

According to Definition \ref{def:ISR-Games}, the cost function of ISR games characterizes industrial symbiotic relations by associating to each singleton group $\{i\}$\footnote{In further references, whenever it is clear from the context that we are referring  to  a singleton group $\{i\} \subset N$, we use $i$ instead of $\{i\}$.}, the total cost that they will face (individually) in case of no-cooperation. Moreover, it ascribes the total ISR-related operational cost to the two member group $N$ as the amount that members of $N$ have to pay (collectively) in case of cooperation. In addition,  the assumed subadditivity of ISR games reflects the feasibility of ISRs, i.e., in ISR games we have that $T(\sigma) \leq \sum\limits_{i \in N} T_i(\bar{\sigma})$. So, regardless of the mechanisms for the allocation of ISR-related operational costs, the total amount to be paid in case of cooperation is at most equal to the sum of the amounts to be paid individually. The following property shows that in general, ISR games are submodular regardless of their particular settings. We later discuss that such a property results in applicability of a large class of game-theoretical cost allocation mechanisms, i.e., mechanisms that are based on the concept of \emph{core of the game}. 

\begin{proposition}[Submodularity of ISR Games] \label{prop:sub2sub} Let $\sigma$ be an arbitrary ISR game. It always holds that $\sigma$ is submodular.\end{proposition}

\begin{proof} According to Definition \ref{def:subs}, a game is submodular iff the  $c(S)+c(T) \geq c(S\cup T) + c(S\cap T)$  inequality holds for all possible agent groups $S$ and $T$ in $N$. In ISR games, by checking the validity of this inequality for all $6$ possible combinations of agent groups (for $S$ and $T$), the claim will be proved. For $S=\emptyset$, we have the following valid inequality $c(\emptyset)+c(T) \geq c(\emptyset \cup T = T) + c(\emptyset)$. For $S=N$, the inequality can be reformulated in the following form that always holds $c(N)+c(T) \geq c(N\cup T=N) + c(N\cap T=T)$. Finally, when $S$  and $T$ are equal to the only two possible  disjoint  groups, we have the following inequality $c(S)+c(T) \geq c(N) + c(\emptyset)$. This inequality always holds thanks to the subadditivity of ISR games. \end{proof}

Note that \emph{submodularity} is not a general property of  \emph{subadditive} cooperative games but holds for the class of ISR games. In the following, we recall our ISR scenario between the glass manufacturer firm $A$ and the ceramics manufacturer firm $B$, and describe the game-theoretical  formulation of this scenario.

\begin{example}[ISR Scenario as a Game] \label{ex:001} In the ISR scenario from section \ref{sec:analysis}, we assume that the amount of recycled excess $r$ in $A$ completely substitutes the required amount of $i$ in $B$. Hence, in case the firms operationalize this symbiotic relation, neither of the firms has to deal with associated traditional costs for discharging excess $r$ and purchasing required $i$. This ISR scenario can be modeled as cooperative game $\sigma=(N,c)$ where $N=\{A,B\}$, $c(A)=7$, $c(B)=11$, $c(\emptyset)=0$, and $c(N)=15$. This game is both subadditive and submodular. To check subadditivity, we survey all possible couples of agent groups in $N$. The only two disjoint agent groups are $\{A\}$ and $\{B\}$ for which the cost of the union group (c($\{A,B\})=c(N)=15$) is lower than the summation of the individual costs ($c(\{A\})+c(\{B\})=18$). Thus, the game $\sigma$ is subadditive. For submodularity, we can rely on Proposition  \ref{prop:sub2sub}. \end{example}

\section{Allocation Mechanisms and Decision Support} \label{sec:mechanism}

Having Industrial Symbiotic Relations (ISRs) modeled as cooperative games (ISR games), one decision that firms are faced with is either ``to get engaged in'' or ``to reject'' a potential ISR. This is mainly to determine  the  \emph{collaboration decision} (as discussed in  Section \ref{sec:2A}). Following rational decision-making perspectives presented in \cite[page.~210-211]{sandholm1999distributed}, we deem that to support the collaboration decision in ISRs, analyzing two key aspects of collaborations, namely  \emph{stability} and \emph{fairness}, results in practical decision support tools for ISRs. In other words, ``goodness'' of a collaboration can be characterized and validated by checking if it is \emph{stable}, \emph{fair}, or both. In the following, we introduce  two methods from game-theoretical literature that axiomatize the \emph{stability} and \emph{fairness} properties. These methods formulate both the stability and fairness of collaborations with respect to distribution of the value that agents can gain thanks to the collaboration. Accordingly, tailoring these mechanisms for ISR games leads to tools for supporting the \emph{collaboration decision} in industrial symbiotic relations.  In this case, the way that industrial agents allocate the total ISR-related operational cost specifies the stability and fairness of the ISR. 

\subsection{Core Allocations for ISR Games}

Core-based mechanisms in cooperative games are mainly concerned about \emph{stability} of possible collaborations \cite{driessen2013cooperative}. In relation to cost allocation in games, a collaboration is stable iff (1) the summation of allocated costs to individual agents in  a collaborative group is equal to the total cost that the group should pay (efficiency) and (2) the allocated costs to individuals is at most equal to their costs in case they defect from the collaborative group (rationality). It is observable that if the agent groups follow a cost allocation method that does not satisfy the two above properties, they end in an \emph{unstable} situation either due to inefficient distribution of costs or as the result of (rational) agents leaving the group. In the following, we define our  core-based cost allocation mechanism for ISR games and describe its  properties.

\begin{definition} [Core Allocations for ISRs] \label{def:core} Let $\sigma$ be an ISR game (as defined in Definition \ref{def:ISR-Games}) between firms $A$ and $B$. The core of $\sigma$ is the set $ \Psi (\sigma) := \{  \langle  T_A^\Psi(\sigma), T_B^\Psi(\sigma)   \rangle \}$ such that for $i \in \{A,B\}    $ we have that (1)  $T_i^\Psi(\sigma) \in \mathbb{R}_{\geq 0}$ (non-negative valued), (2) $\sum\limits_{i \in \{A,B\}} T_i^\Psi(\sigma) = T(\sigma)$ (efficient), and (3) $T_i^\Psi(\sigma) \leq T_i(\bar{\sigma})$ (individually rational). \end{definition}

Following the discussion about stability of collaborations, an ISR $\sigma$ is stable with respect to the allocation of its operationalization costs, if it implements a cost allocation that belongs to the core of $\sigma$. Thus, the presented core allocation for ISRs can be applied as (1) a mechanism for guaranteeing the stability of an ISR (ISR design) and (2) as a verification method to analyze if an ISR is stable (ISR assessment). The set of core allocations for a given ISR, construct a segment representable in two-dimensional space (see Figure \ref{fig:M1}). Due to the linear formulation  of ISRs' core, computing the set of stable allocations is not computationally expensive and can be computed applying linear programming techniques to solve a  two-variable system of inequalities \cite{dantzig2016linear}.

\begin{figure}[tb]
\centering
\begin{tikzpicture}[
    scale=2.5,
    axis/.style={thick, ->, >=stealth},
    important line/.style={thick},
    dashed line/.style={dashed, thin},
    pile/.style={line width=.8mm},
    every node/.style={color=black}
    ]
    \draw[axis] (0,0)  -- (1.8,0) node(xline)[right]
        {$T_A(\sigma)$};
    \draw[axis] (0,0) -- (0,1.8) node(yline)[above] {$T_B(\sigma)$};
    \draw[important line] (0,1.5) coordinate (A) node[left, text width=2em] {$T(\sigma)$} -- (1.5,0)         coordinate (B) node[below, text width=2em] {$T(\sigma)$};
    \draw[pile] (.4,1.1) coordinate (C) node [above,  yshift=-0.05cm, xshift=0.05cm] {$\alpha$} -- (.55,.95)
        coordinate (D) node [left, yshift=-0.1cm, xshift=0.05cm] {$\gamma$};

    \draw[pile] (.55,.95) coordinate (V) node [above] {} -- (.7,.8)
        coordinate (Y) node [right, yshift=0.1cm, xshift=-0.05cm] {$\beta$};
\fill (D) circle[radius=1pt];

\draw[dashed line] (.4,1.1) coordinate (C1) -- (0,1.1)
        coordinate (D1) node [left, text width=2.5em] {$T_B(\bar{\sigma})$};      
        
\draw[dashed line] (.4,1.1) coordinate (C2) -- (.4,0)
        coordinate (D2) node [below, text width=4em] {$U_A(\sigma)$};
                
\draw[dashed line] (.7,.8) coordinate (C3) -- (.7,0)
        coordinate (D3) node [below, text width=2em] {$T_A(\bar{\sigma})$};
        
\draw[dashed line] (.7,.8) coordinate (C4) -- (0,.8)
        coordinate (C4) node [left, text width=2.5em] {$U_B(\sigma)$};
        
\draw[decorate,decoration={brace,amplitude=5pt,raise=0.3pt},yshift=0pt] (C) -- (Y) node [midway,yshift=10pt, xshift =45pt]{$\{\langle  T_A^\Psi(\sigma), T_B^\Psi(\sigma)   \rangle \}$};
   
\end{tikzpicture}
\caption{Schematic core and Shapley allocations for ISR game $\sigma$: In this diagram, allocated costs to firms $A$ and $B$ are illustrated on the horizontal and vertical axes, respectively. Moreover, $U_A(\sigma)=T(\sigma)-T_B(\bar{\sigma})$, $U_B(\sigma)=T(\sigma)-T_A(\bar{\sigma})$, and $\gamma$ represents the Shapley allocation of the game $\Phi (\sigma)$.} \label{fig:M1}
\end{figure}
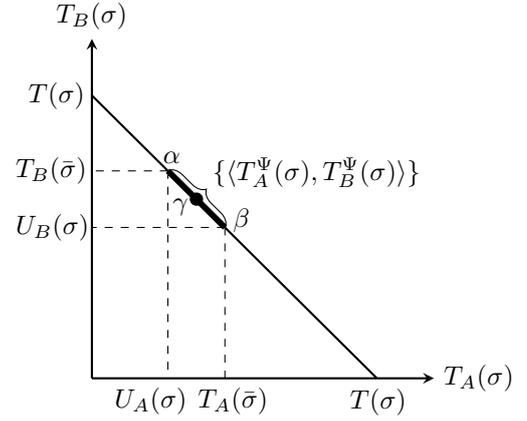

Two main concerns with respect to  applicability of the concept of \emph{core} are  about its \emph{existence} and \emph{fairness}. In other words, is the set of core allocations a nonempty one and does it always provide a fair distribution of costs among involved firms in an ISR game?  In the following, we first illustrate  that ISR games have a nonempty core and  core allocations cannot guarantee the fairness property. This motivates the introduction of a fair cost allocation mechanism, i.e., a Shapley-based allocation mechanism for ISR games. 

\begin{proposition}[Existence Property of Core of ISR Games] \label{prop:core_nonempty}
Let $\Psi (\sigma)$ be the core of ISR game $\sigma$ between firms $A$ and $B$. It always holds that $\Psi (\sigma) \neq \emptyset$.
\end{proposition}

\begin{proof} According to the \emph{Bondareva-–Shapley} theorem (as described in \cite{driessen2013cooperative}), submodular games have a nonempty core. We already proved in Proposition  \ref{prop:sub2sub} that ISR games are submodular. Hence, we have that for any ISR game $\sigma$, the core is not empty.\end{proof}

Note that the concept of \emph{core of  the game}, provides a set of allocations that guarantee the stability of the collaboration. Having a range of stable cost distributions is  appropriate for ISR scenarios in which firms are allowed to practice their bargaining power (see \cite{coff1999competitive}) in the negotiation process in order to pay the smallest possible share of the ISR-related operational costs. E.g., as illustrated in Figure \ref{fig:M1}, the most desirable (albeit stable) allocation form for firms $A$ and $B$ occurs in points $\alpha$ and $\beta$, respectively. For instance in $\alpha$, firm $A$ enjoys paying the lowest stable share of the total ISR-related operational cost while $B$ is paying the highest. In this case, $B$ suffers because of this intuitively unfair allocation and pays equal to its total traditional cost. This shows that the concept of core provides  a method to verify the stability of an ISR and can be applied as a tool to support the \emph{collaboration decision} in ISRs, i.e., accepting an out of core cost-distribution results in an \emph{unstable}  ISR. Nevertheless,  it does not grasp the fairness and neither provides a method for choosing a specific cost allocation to implement among the set of stable cost allocations. In the following, we tailor a solution concept that axiomatize the concept of fairness and provides a  single-point allocation that satisfies the fairness property. 

\subsection{Shapley Allocation for ISR Games}

Regarding the allocation of costs in collaborative groups, there exist various interpretations of the complex notion of fairness (see \cite{rosenbusch2011fairness}). In cooperative game theory, the well-established concept of \emph{Shapley value} \cite{shapley1953value} is a central concept that regards the \emph{fairness} of a cost distribution among members of a collaborative agent group by taking into account their contributions  to the collaborative group. In this work, we follow Shapley's view and expect a \emph{fair} allocation to satisfy the following  four properties: \emph{efficiency}, \emph{symmetry}, \emph{dummy player}, and \emph{additivity} (as discussed in Section \ref{sec:GTanalysis}). In brief, if a group $G$ follows an \emph{efficient} method to allocate cost $C$ among its members, the summation of allocated costs to members of group $G$ will be equal to $C$. The \emph{symmetry} property says that agents that make the same contribution to the total cost, should be allocated the same individual cost shares. The \emph{dummy player} property says that if the presence of an agent $A$ does not result in any cost reduction (in all the possible agent groups), the allocated cost to $A$ should be equal to its individual total traditional cost. Finally, the \emph{additivity} property says that if you combine two games $V$ and $U$, the allocated cost to an agent $A$ (involved in both the games) should be the sum of the costs allocated to $A$ in the individual games, i.e., playing more than once does not lead to any (dis)advantages for $A$. For formal axiomatization of these properties, we refer the reader to \cite{shapley1953value}. In the following, we present our tailored Shapley value for ISR games.

\begin{definition} [Shapley Allocation for ISRs] \label{def:Shapley} Let $\sigma$ be an ISR game (as defined in Definition \ref{def:ISR-Games}) between firms $A$ and $B$. The Shapley allocation for $\sigma$ is the tuple $ \Phi (\sigma) :=   \langle  T_A^\Phi(\sigma), T_B^\Phi(\sigma)   \rangle $ where for $i \in N=\{A,B\}$ we have $T_i^\Phi(\sigma) = \frac{1}{2} [T(\sigma)+T_i(\bar{\sigma})-T_{N \setminus \{i\}}(\bar{\sigma})].$ \end{definition}

A reader familiar with the notion of Shapley value might expect the two notions of \emph{orders} and \emph{marginal contributions} to be a part of our tailored concept of Shapley value for ISRs. We highlight that due to our domain of application, i.e, bilateral industrial relations, there are two possible orders in ISR games (reflected by the constant value $\frac{1}{2}$). Moreover, the marginal contribution of a given firm $i \in N=\{A,B\}$ can be reformulated in terms of the most desirable stable cost for $i$ and the most undesirable one, i.e., $U_i(\sigma)=T(\sigma)-T_{N \setminus \{i\}}(\bar{\sigma})$ and $T_i(\bar{\sigma})$, respectively. Note that as the Shapley value is defined following a constructive method (in contrast to condition-based definition of core in Definition \ref{def:core}), the existence of the  Shapley value for any arbitrary ISR  game is guaranteed. Following our discussion about the fairness of collaborations, the next property shows that the Shapley value is the unique fair method for allocation of the total ISR-related operational cost in ISR games.

\begin{proposition} [Uniqueness of the Shapley Value] \label{prop:uni} Let $\Phi (\sigma)$ be the Shapley allocation for ISR game  $\sigma$ between firms $A$ and $B$. For any fair allocation of costs in $\sigma$, denoted by $\Phi'(\sigma)$, we have that $\Phi' (\sigma) = \Phi(\sigma)$.\end{proposition}

\begin{proof} Importing results from \cite{shapley1953value}, we have that for any cooperative game, the Shapley value is the unique allocation method that satisfies all the four properties of fair cost allocations, i.e., efficiency, symmetry, dummy player, and additivity, regardless of the characteristics of the cost function of the game. Accordingly, the uniqueness property holds for ISR games as two-person cost allocation games. \end{proof}

Considering core of ISR games and their unique Shapley allocation, the following proposition relates these two forms of solution concepts and shows that in ISR games the Shapley allocation is in the core. 

\begin{proposition} [Membership in the Core] \label{prop:membership} Let $\Psi (\sigma)$ and $\Phi (\sigma)$ be respectively the set of core allocations, and the Shapley allocation for ISR game $\sigma$ between firms $A$ and $B$. It always holds that $\Phi (\sigma) \in \Psi (\sigma)$.\end{proposition}

\begin{proof} Based on \cite{shapley1953value,young1985cost}, the core of submodular games is nonempty and includes the Shapley value. For ISR games, according to Proposition \ref{prop:core_nonempty}, the core is nonempty. Thus, we have that for any ISR game $\sigma$, it holds that $\Phi (\sigma) \in \Psi (\sigma)$. \end{proof}

Note that the membership of the Shapley allocation in the core  is a property of ISR games and not a general property of the Shapley cost allocation for any class of cooperative games. Accordingly, we have that the Shapley allocation of any ISR game can be illustrated in two dimensional space. More precisely, the Shapley allocation $\gamma$ (see Figure \ref{fig:M1}) is the midpoint of the core allocation segment.

\begin{example}[Allocations in the ISR Scenario] \label{ex:002} Considering the presented scenario in Example \ref{ex:001}, any cost allocation  $\langle  T_A, T_B  \rangle$ such that $4 \leq T_A \leq 7$ and    $8 \leq T_B \leq 11$ and $T_A + T_B =15$ is a core allocation. Moreover, $\langle 5.5 , 9.5  \rangle$ is the Shapley allocation.
\end{example}

In general, Shapley allocation does not satisfy the \emph{individual rationality} property. I.e., it might be the case that some agents in a collaborative group should pay higher than their traditional cost in order to guarantee the fairness property. In such cases, we have a fair albeit unstable collaboration because any sacrificing firm has incentive to rationally defect the collaboration. However, the next proposition shows that for ISR games, the Shapley allocation is individually rational. 

\begin{proposition} [Fairness and Stability] \label{prop:raitonalToo}
Let $\Phi (\sigma)$ be the Shapley allocation for ISR game $\sigma$ between firms $A$ and $B$. It always holds that for $i \in \{A,B\}$, we have that $T_i^\Phi (\sigma) \leq T_i(\bar{\sigma})$.\end{proposition}

\begin{proof} According to Proposition  \ref{prop:membership}, the Shapley allocation of ISR game $\sigma$ is also a core allocation. Hence, it also satisfies the \emph{individual rationality} condition in Definition   \ref{def:core}. \end{proof}

Based on proposition \ref{prop:raitonalToo}, in case the two firms agree to implement the Shapley allocation, it is guaranteed that the relation will be both fair and stable. Although implementing the Shapley allocation seems natural due to its desirable  properties, firms may prefer to negotiate among cost allocations in the core with the aim to practice their bargaining power and enjoy more cost reduction. However, industrial agents that suffer from unfair allocations in such cases may defect the collaboration, leave/reject the  ISR, and join other ISRs that are practicing fair allocation methods. 

\section{Conclusions and Future Work} \label{sec:conc}

In this work, we present a game-theoretical representation of Industrial Symbiotic Relations (ISRs) and tailor two types of solution concepts for cost allocation in such relations, i.e., \emph{core allocation} and \emph{Shapley allocation} for ISR games. These two notions can also be seen as two approaches for decision support while firms are faced with the collaboration decision, to reject or accept an ISR proposal. This is by enabling firms to systematically reason about and verify \emph{stability} and \emph{fairness} of a particular ISR. Range of stable collaborations, provided by the concept of \emph{core}, allows further negotiation while the \emph{Shapley} allocation leads to a uniquely fair solution. We then show that due to the characteristics of industrial symbiotic relations, ISR games can  always be operationalized in both a fair and stable manner. In addition to practical contributions by providing managerial decision support tools, we introduced ISR games as a new class of two-person Operations Research (OR) games. In ISR games, we have the non-emptiness of the core and it is guaranteed that the Shapley value in this class of OR games is an individually rational solution. As a future work, we aim to analyze the validity of presented results using  multiagent-based simulations \cite{dignum2009agents}. We also plan to extend our game-theoretical analysis to network level, relate our notions to the concept of \emph{willingness to cooperate} \cite{yazan2012appendix}, and study Industrial Symbiotic Networks (ISNs). Due to complexities of ISNs (see \cite{yazdanpanah2016}) guaranteeing \emph{fairness} and \emph{stability} in such networks calls for mechanisms to coordinate agent interactions \cite{dastani2015commitments} and governance platforms for, as discussed in  \cite{singh2013norms}, ``administration of stakeholders by stakeholders".

\section*{Acknowledgment}
The project leading to this work  has received funding from the \emph{European Union's Horizon
2020 research and innovation programme} under grant agreement \emph{No. 680843}.

\bibliographystyle{IEEEtran}
\bibliography{Main}
\end{document}